\documentclass{ws-procs9x6}

\usepackage{cite}     

\begin{document}

\date{}
\title{On a generalization of the Dirac bracket in the 
De Donder-Weyl Hamiltonian formalism\footnote{
Proc. 10th Int. Conf. on Differential Geometry  and its  Applications 
DGA 2007, Olomouc, Czech Republic.}
}

\author{I. KANATCHIKOV$^*$}

\address{Institute of Theoretical Physics,
TU Berlin,  \\ 
D-10623 Berlin, Germany \\   
$^*$E-mail: ivar@itp.physik.tu-berlin.de
}

\begin{abstract}
The elements of the contrained dynamics algorithm 
in the De Donder-Weyl (DW) Hamiltonian theory  for degenerate 
Lagrangian theories 
are discussed. A generalization of the Dirac bracket to 
the DW Hamiltonian theory 
with second class constraints (defined in the text) is presented.   
\end{abstract}



\bodymatter


\newcommand{\beq}{\begin{equation}}
\newcommand{\eeq}{\end{equation}}
\newcommand{\beqa}{\begin{eqnarray}}
\newcommand{\eeqa}{\end{eqnarray}}
\newcommand{\nn}{\nonumber}

\newcommand{\half}{\frac{1}{2}}

\newcommand{\xt}{\tilde{X}}

\newcommand{\uind}[2]{^{#1_1 \, ... \, #1_{#2}} }
\newcommand{\lind}[2]{_{#1_1 \, ... \, #1_{#2}} }
\newcommand{\com}[2]{[#1,#2]_{-}} 
\newcommand{\acom}[2]{[#1,#2]_{+}} 
\newcommand{\compm}[2]{[#1,#2]_{\pm}}

\newcommand{\lie}[1]{\pounds_{#1}}
\newcommand{\co}{\circ}
\newcommand{\sgn}[1]{(-1)^{#1}}
\newcommand{\lbr}[2]{ [ \hspace*{-1.5pt} [ #1 , #2 ] \hspace*{-1.5pt} ] }
\newcommand{\lbrpm}[2]{ [ \hspace*{-1.5pt} [ #1 , #2 ] \hspace*{-1.5pt}
 ]_{\pm} }
\newcommand{\lbrp}[2]{ [ \hspace*{-1.5pt} [ #1 , #2 ] \hspace*{-1.5pt} ]_+ }
\newcommand{\lbrm}[2]{ [ \hspace*{-1.5pt} [ #1 , #2 ] \hspace*{-1.5pt} ]_- }

\newcommand{\pbr}[2]{ \{ \hspace*{-2.2pt} [ #1 , #2\hspace*{1.4 pt} ] 
\hspace*{-2.3pt} \} }
\newcommand{\nbr}[2]{ [ \hspace*{-1.5pt} [ #1 , #2 \hspace*{0.pt} ] 
\hspace*{-1.3pt} ] }

\newcommand{\we}{\wedge}
\newcommand{\nbrpq}[2]{\nbr{\xxi{#1}{1}}{\xxi{#2}{2}}}
\newcommand{\lieni}[2]{$\pounds$${}_{\stackrel{#1}{X}_{#2}}$  }

\newcommand{\rbox}[2]{\raisebox{#1}{#2}}
\newcommand{\xx}[1]{\raisebox{1pt}{$\stackrel{#1}{X}$}}
\newcommand{\xxi}[2]{\raisebox{1pt}{$\stackrel{#1}{X}$$_{#2}$}}
\newcommand{\ff}[1]{\raisebox{1pt}{$\stackrel{#1}{F}$}}
\newcommand{\dd}[1]{\raisebox{1pt}{$\stackrel{#1}{D}$}}
\newcommand{\der}{\partial}
\newcommand{\oo}{$\Omega$}
\newcommand{\Om}{\Omega}
\newcommand{\om}{\omega}
\newcommand{\eps}{\epsilon}
\newcommand{\si}{\sigma}
\newcommand{\Lm}{\bigwedge^*}

\newcommand{\inn}{\hspace*{2pt}\raisebox{-1pt}{\rule{6pt}{.3pt}\hspace*
{0pt}\rule{.3pt}{8pt}\hspace*{3pt}}}
\newcommand{\sro}{Schr\"{o}dinger\ }

\renewcommand{\bm}{\boldmath} 

\newcommand{\vol}{\omega}
               \newcommand{\dvol}[1]{\der_{#1}\inn \vol}

\newcommand{\bd}{\mbox{\bf d}}
\newcommand{\bder}{\mbox{\bm $\der$}}
\newcommand{\bI}{\mbox{\bm $I$}}

\newcommand{\be}{\beta} 
\newcommand{\ga}{\gamma} 
\newcommand{\de}{\delta} 
\newcommand{\Ga}{\Gamma} 
\newcommand{\gmu}{\gamma^\mu}
\newcommand{\gnu}{\gamma^\nu}
\newcommand{\ka}{\kappa}
\newcommand{\hka}{\hbar \kappa}
\newcommand{\al}{\alpha}
\newcommand{\lapl}{\bigtriangleup}
\newcommand{\psib}{\overline{\psi}}
\newcommand{\Psib}{\overline{\Psi}}
\newcommand{\Phib}{\overline{\Phi}}
\newcommand{\derts}{\stackrel{\leftrightarrow}{\der}}
\newcommand{\what}[1]{\widehat{#1}}
\newcommand{\pib}{\overline{\pi}}
\newcommand{\Cb}{\overline{C}}
\newcommand{\ub}{\overline{u}}

\newcommand{\bx}{{\bf x}}
\newcommand{\bk}{{\bf k}}
\newcommand{\bq}{{\bf q}}

\newcommand{\omk}{\omega_{\bf k}} 
\newcommand{\lpl}{\ell}
\newcommand{\zb}{\overline{z}} 

\newcommand{\dv}{\mbox{\sf d}}

\newcommand{\BPsi}{{\bf \Psi}} 
\newcommand{\BH}{{\bf H}} 
\newcommand{\BS}{{\bf S}} 
\newcommand{\BN}{{\bf N}} 


\section{Introduction}\label{aba:sec1}

In Refs.\cite{ikanat93,romp98,bial96,goslar21} a 
generalization of the Poisson bracket to the De Donder-Weyl (DW) Hamiltonian 
formulation \cite{rund,gimm} in field theory 
(and the multiparametric calculus of variations)  
has been proposed. It leads to a Gerstenhaber algerba of brackets defined
on specific horizontal differential forms (called Hamiltonian forms).  
The construction  assumes that the corresponding Lagrangian theory 
is regular (in the sense of DW Hamiltonian formulation):
\beq \label{regularity}
\det \left |\left | \frac{\der^2 L}{\der \phi_a^\mu\der \phi_a^\nu }
\right |\right| \neq 0 . 
\eeq
However, many field theories of physical interests, such as spinor fields, 
are not regular in the above sense. 
In this paper  we discuss a corresponding 
extension of the formalism  to the degenerate case 
and generalize  
the above mentioned Poisson-Gerstenhaber bracket 
similarly to the construction of the Dirac bracket in the constrained systems 
with second-class constraints \cite{henneaux}. 

Note that the regularity condition in the DW Hamiltonian formalism 
is different from the one in the standard Hamiltonian formalism: 
$\det \left |\left | {\der^2 L}/{\der \dot{\phi}_a\der \dot{\phi}_a }
\right |\right| \neq 0 $, which involves only the time derivatives 
of fields $\dot{\phi}_a$. 
As a consequence, the theories which are regular from the point of view 
of the DW formalism can be irregular from the point of view of the 
standard Hamiltonian formalism and {\em vice versa}. 

This  work is motivated by  the  project of 
manifestly space-time symmetric {\em precanonical} quantization 
of field theory \cite{bial97,ijtp2001,opava2001,ym} 
based on the DW Hamiltonian formalism, which requires a well 
understood formalism  for degenerate theories on the classical level. 

\section{The polysymplectic structure and the Poisson-Gerstenhaber brackets} 

In this section we briefly recall the construction and properties of 
the Poisson-Gerstenhaber brackets in DW Hamiltonian formalism 
which are used in the following section.  
For further generalizations and a detailed geometrical treatment 
of the related issues I refer to the existing literature 
\cite{romanroy,francav,paufler,deleon,krupkova,sardan}. Here I closely follow 
my previous work \cite{romp98,bial96,goslar21,opava2001}.

The DW Hamiltonian formalism \cite{rund,gimm} is a 
space-time symmetric generalization of the Hamiltonian formalism 
from mechanics to field theory. Given the first order Lagrangian density 
$L = L(y^a, y^a_\mu, x^\nu)$ we introduce the variables  (polymomenta)  
\[p^\mu_a := {\der L}/{\der y^a_\mu}\] 
and the DW Hamiltonian function 
\[H:= y^a_\mu p^\mu_a - L = H(y^a, p^\mu_a, x^\mu).\]
Then the Euler-Lagrange equations take the form
\beq \label{dw}
\der_\mu y^a (x) = {\der H}/{\der p^\mu_a},  
\quad 
\der_\mu p^\mu_a (x) = - {\der H}/{\der y^a} ,  
\eeq 
which is equivalent to the Euler-Lagrange field equations if 
$L$ is {regular} in the sense of Eq. (\ref{regularity}). 

Geometrically, classical  fields $y^a= y^a(x)$ are sections in the 
{\em covariant configuration bundle} $Y\rightarrow X$ 
over an oriented $n$-dimensional space-time manifold $X$ 
with the volume form 
$\omega$. The local coordinates in 
 $Y$ are $(y^a,x^\mu)$. 
If 
$\bigwedge{}^p_q(Y)$ denotes the space of 
$p$-forms on $Y$ which are annihilated by $(q+1)$ arbitrary vertical 
vectors of $Y$, then $\bigwedge^n_1(Y)\rightarrow Y$     
generalizes the cotangent bundle  and serves as a  model of 
 a 
{\em multisymplectic phase space\/} \cite{gimm} 
which possesses the canonical $n$-form structure 
\beq 
\Theta_{MS} = p_a^\mu dy^a \we \omega_\mu + p\, \omega 
\eeq
called multisymplectic, where 
$\omega_\mu := \der_\mu\inn\omega$ 
are the basis of $\bigwedge^{n-1} T^*X$.
A section $p= - H(y^a,p^\mu_a,x^\nu)$ yields the 
multidimensional 
{\em Hamiltonian Poincar\'e-Cartan form\/}  $\Theta_{PC}$.  

In order to introduce the Poisson brackets 
which reflect the dynamical structure of DW Hamiltonian 
equations (\ref{dw}) we need a structure which is independent 
of $p$ or a choice of $H$. We define the 
{\em extended polymomentum phase space\/} as the quotient bundle 
$Z$: $\bigwedge^{n}_1(Y) / \bigwedge^{n}_0(Y) \!\rightarrow \! Y.  
$ 
The local coordinates on $Z$ are $(y^a,p_a^\nu,x^\nu)=: (z^v, x^\mu)=z^M$. 
We introduce a  canonical structure on $Z$ 
 as an equivalence class 
of forms 
$\Theta := [p_a^\mu dy^a\we \omega_\mu \quad {\rm mod} \bigwedge{}^{n}_0(Y)]$. 
The {\em polysymplectic structure\/} on $Z$ is defined as  
an equivalence class 
 %
of 
forms $\Omega$ given by 
\beq
\Omega := [d\Theta \quad {\rm mod} \;\mbox{$\bigwedge^{n+1}_1(Y)$}] 
= [- dy^a\we dp^\mu_a \we \omega_\mu \quad {\rm mod}  
\;\mbox{$\bigwedge^{n+1}_1(Y)$}] . 
\eeq 
The equivalence classes are introduced as an alternative to 
the explicit introduction of a non-canonical connection 
on the multisymplectic phase space in order to define 
the polysymplectic structure as a ``vertical part'' of the 
multisymplectic form $d\Theta_{MS}$. 
The fundamental constructions, 
such as the Poisson bracket below,  are 
designed to be independent of the choice of representatives in the 
equivalence classes 
and the choice of a connection. 

A multivector field of degree $p$, $\xx{p}\in \bigwedge^p TZ$,  
is called {\em vertical\/} 
if $\xx{p}\inn F  = 0$ for any form $ F \in \bigwedge^{*}_0(Z)$. 
The polysymplectic form establishes a map of 
horizontal $p-$forms  
$\ff{p}$$\in$$\bigwedge^p_0(Z)$, $p=0,1,..., (n-1)$,
to vertical multivector fields 
of degree $(n-p)$, $\xx{n-p}{}_F$,  called {\em Hamiltonian\/}: 
\beq \label{map}
\xx{n-p}{}_F\inn \Omega = d \ff{p}.  
\eeq 
The forms for which the map (\ref{map}) exists are also 
called {\em Hamiltonian\/}. 
More precisely, 
horizontal  forms are mapped to the {\em equivalence 
classes\/} of Hamiltonian multivector fields 
modulo the {\em characteristic\/}   
multivector fields $\xx{p}_0$: $\xx{p}_0\inn\Omega = 0$, $p=2,...,n$. 
It is important to note that  the space of Hamiltonian forms is not stable 
with respect to the exterior product.  The natural product operation of 
Hamiltonian forms is the {\em co-exterior\/} product 
 \beq
\ff{p}\bullet \ff{q} := *^{-1}(*\ff{p}\we *\ff{q})  
\eeq 
which is graded commutative and associative.

The Poisson bracket of Hamiltonian forms is given by the formula 
\beq \label{pbrformula}
\pbr{\ff{p}{}_1}{\ff{q}{}_2} = (-1)^{(n-p)} \xx{n-p}{}_1 \inn d \ff{q}{}_2 
= (-1)^{(n-p)} \xx{n-p}{}_1 \inn \xx{n-q}{}_2 \inn \Omega . 
\eeq 
In fact, it is  induced by the 
 {Schouten-Nijenhuis} bracket $\nbr{\;}{\,}$ 
of the corresponding 
Hamiltonian multivector fields: 
\beq
- d \pbr{\ff{p}}{\ff{q}} := 
\nbr{\xx{n-p}}{\xx{n-q}} \inn \Omega . 
\eeq
The algebraic properties of the bracket 
 are summarized  in the following 
\begin{theorem} \label{the1} 
The space of Hamiltonian forms with the 
operations $\pbr{\;}{\,}$ and $\bullet$  
is  a {\bf (Poisson-)Gerstenhaber algebra\/}, i.e. 
\beqa
\pbr{\ff{p} }{\ff{q} } &=& -(-1)^{g_1 g_2}
\pbr{\ff{q}}{\ff{p}}, \nn \\ 
\mbox{$(-1)^{g_1 g_3} \pbr{\ff{p}}{\pbr{\ff{q}}{\ff{r}}}$} 
&\!+\!& 
\mbox{$(-1)^{g_1 g_2} \pbr{\ff{q}}{\pbr{\ff{r}}{\ff{p}}}$} 
  \\
&& \hspace*{15pt}+ \quad \! 
\mbox{$(-1)^{g_2 g_3} \pbr{\ff{r}}{\pbr{\ff{p}}{\ff{q}}} 
 \; = \;0, $}  
 \nn \\
\pbr{\ff{p}}{\ff{q}\bullet \ff{r}} 
&=& 
\pbr{\ff{p}}{\ff{q}}\bullet \ff{r}
+ (-1)^{g_1(g_2+1)} \ff{q}\bullet\pbr{\ff{p}}{\ff{r}},   
\nn 
\eeqa 
where $g_1 = n-p-1$,  $g_2 = n-q-1$,  $g_3 = n-r-1$.  
\end{theorem} 
The graded Lie algebra properties are induces by the graded Lie 
properties of the 
Schouten-Nijenhuis bracket. The graded Leibniz property is a consequence 
of the Fr\"olicher-Nijenhuis classification of graded derivations 
adapted to the co-exterior  algebra of forms.


Let us mention a few applications of 
the Poisson-Gerstenhaber (PG) brackets. They enable us to identify the pairs of 
canonically conjugate variables which may become   
a starting point of quantization: 
\beq
\pbr{p_a^\mu\omega_\mu}{y^b}
= 
\delta^b_a , \quad 
\pbr{p_a^\mu\omega_\mu}{y^b\omega_\nu}
=
\delta^b_a\omega_\nu, \quad 
\pbr{p_a^\mu}{y^b\omega_\nu}
= 
\delta^b_a\delta^\mu_\nu .   
\refstepcounter{equation} 
\eeq
In fact, geometric prequantization of PG brackets has been discussed 
in our previous paper \cite{opava2001}. 
The brackets can be used also in order to write the DW 
Hamiltonian equations in the bracket form: 
for any Hamiltonian $(n-1)-$form  $F:=F^\mu\omega_\mu$ 
the equations of motion take the form 
\beq \label{bracketform}
\bd\!\bullet\! F = -\sigma (-1)^{n}\pbr{H}{F}  
+ d^h \!\bullet \! F , 
\eeq
where $\bd\bullet$ denotes the ``total 
co-exterior differential'' 
\beq
\bd\!\bullet\! \ff{p}:= \frac{1}{(n-p)!}\der_M F\uind{\mu}{n-p}\der_\mu z^M 
dx^\mu\bullet \der\lind{\mu}{n-p}\inn\ \omega,  
\eeq 
$d^h$ is the ``horizontal co-exterior differential'': 
\beq
d^h\!\bullet \! \ff{p}:= \frac{1}{(n-p)!}\der_\mu F\uind{\mu}{n-p}   
dx^\mu\bullet \der\lind{\mu}{n-p} \inn\ \omega, 
\eeq
and $\sigma=\pm 1$ for the Euclidean/Minkowskian signature of the 
base manifold $X$.   Consequently, the conservation of the quantity 
represented by $(n-1)$-form $F$ is equivalent to the condition 
$\pbr{H}{F}  = 0$.

\section{The constrained dynamics and the Dirac bracket of forms}

The violation of the regularity condition (\ref{regularity})
implies the noninvertibility of the space-time gradients of fields 
$\phi^a_\nu$ 
as functions of the field variables   $\phi^a$ 
and polymomenta  $p_a^\nu$. 
In this case  the polymomenta  are not all independent,
and there 
exist relations 
\beq \label{eqconstr}
C_A(\phi^a, p_a^\nu) =0  
\eeq 
that follow from the definition  of the polymomenta and the form 
of the Lagrangian. The conditions (\ref{eqconstr}), 
as usual, will be called {\bf primary constraints } 
to emphasize that the field 
equations are not used to obtain these relations and that they imply 
no restriction on the field variables and their space-time gradients. 


The canonical DW Hamiltonian function is now not unique. It can be replaced by 
the effective DW Hamiltonian function which is weakly equal to  $H$: 
\beq
\tilde{H} := H + u_A C_A \approx H .
\eeq
In the present approach we restrict our attention to the  constraints 
which can be organized into 
Hamiltonian $(n-1)-$forms 
\beq
{\mathfrak C_m}= C_m^\mu\om_\mu . 
\eeq
In this case the effective Hamiltonian can be written as 
\beq
\tilde{H} := H + u_m \bullet {\mathfrak C_m},  
\eeq
where the Lagrange multipliers are organized into one-forms $u_m$. 

\begin{definition}
{\it A dynamical variable represented by a semi-basic Hamiltonian form ${F}$ 
of degree $|{F} |=f$  
is said to be {\bf first class } if its PG bracket with every constraints $(n-1)$-form ${\mathfrak{C}_m}$ weakly vanishes, } 
\beq
\pbr{{F}}{\mathfrak{C}_m} \approx 0. 
\eeq 

\end{definition}

\begin{definition}
{\it A semi-basic Hamiltonian form 
${F}$  is said to be {\bf second class } if there is at least one 
constraint 
such that its PG bracket with ${F}$ does not vanish weakly. }
\end{definition}

\medskip

\begin{proposition}
{\it The first-class property is {\bf preserved} under the PG bracket operation, 
i.e. the PG bracket of two  first-class Hamiltonian forms is first class. }
\end{proposition}

\begin{proof}
If ${F}$ and ${G}$ are first class, then 
\beq \label{eqv}
\pbr{{F}}{\mathfrak{C}_m} = f_m^n \bullet \mathfrak{C}_n,  \quad 
\pbr{{G}}{\mathfrak{C}_m}= g_m^n\bullet  \mathfrak{C}_n
\eeq
for some $(f+1)$-forms $f_m^n$ and $(g+1)$-form $g_m^n$, where $|{G} |=g$, $|{F} |=f$.
By the graded Jacobi identity  
\[
-(-1)^{d_G}\pbr{\mathfrak{C}_m}{\pbr{F}{G}} =  
(-1)^{d_F} \pbr{{F}}{\pbr{{G}}{\mathfrak{C}_m}}  
+(-1)^{d_F d_G} \pbr{G}{\pbr{{\mathfrak{C}_m}}{F}},  \]
where $d_G:= n-g-1, \; d_F:=n-f-1$. 
Using Eq. (\ref{eqv}), the graded Leibniz rule, and 
the graded antisymmetry 
of the bracket in the right hand side, we obtain
\beqa
&&(-1)^{d_G}\pbr{{F}}{g_m^n\bullet \mathfrak{C}_n }  
- (-1)^{d_Fd_G +d_F} \pbr{{G}}{f_m^n \bullet\mathfrak{C}_n }  
\approx 0.  \nn
\eeqa 
Thus, 
\[\pbr{\mathfrak{C}_m}{\pbr{F}{G}} \approx 0 \] 
for any first class Hamiltonian forms $F$ and $G$.
\end{proof}


Next, we consider the necessary condition for the 
conservation of 
the constraints ${\mathfrak C }_m forms: \bd \bullet {\mathfrak C}_m = 0$ 
(c.f. Eq. (\ref{bracketform})). 
They give rise to the  necessary 
consistency conditions 
\beq \label{constra}
\pbr{\tilde{H}}{\mathfrak{C}_m}=\pbr{H }{\mathfrak{C}_m} 
+ u^n \bullet \pbr{\mathfrak{C}_n}{\mathfrak{C}_m} \approx 0,
\eeq
where $u^n$ are one-form coefficients. 
Eq. (\ref{constra}) can either  impose a restriction on the $u$'s or it 
may reduce to a new relation independent of the $u$'s.  In the latter case, 
i.e. if the new relation between $p$'s and $\phi$'s  is independent of 
the primary constraints, 
it will be called a {\bf secondary constraint}, 
in accordance with  the 
conventional terminology introduced by Dirac and Bergmann. 

While the primary constraints are merely consequences of the definition of the polymomentum variables, 
the secondary constraints make use of the equations of motion 
(the DW Hamiltonian field equations) as well. 
If there is a secondary constraint written as an $(n-1)$-form  
${\mathfrak B}_m (\phi, p):=B_m^\mu\om_\mu$ appearing, 
a new consistency condition  in the form (\ref{constra})
\beq
\pbr{H}{{\mathfrak B}_m} + u^n \bullet\pbr{{\mathfrak B}_n}{\mathfrak{C}_m} \approx 0
\eeq
must be imposed. 
 Next, we must again check whether Eq. (22) implies new secondary 
constraints or whether it 
only restructures the $u$'s, and so on. 
It is similar to the standard procedure originally proposed by Dirac. 


 The {\bf second-class constraints } are present 
whenever the $(n-1)-$form valued matrix 
\beq
\mathfrak{C}_{mn} := \pbr{\mathfrak{C}_m}{\mathfrak{C}_n}
\eeq
does not vanish on the constraints surface. As usual, when computing ${\mathfrak C}_{mn}$  we  must not use the constraints equations until after calculating the 
Poisson-Gerstenhaber bracket. 
For simplicity, we assume that the constraints are irreducible and 
the rank of 
 ${\mathfrak C}_{mn}$  
is constant on the constraints surface. 
Since ${\mathfrak C}_{mn}$ is a nonsingular matrix whose components are $(n-1)$-forms, 
 its ``inverse'' ${\mathfrak C}^{-1}_{mn}$ exists such that 
\beq
 {\mathfrak C}^{-1}_{mn}\we {\mathfrak C}_{nk} = \delta_{mk}\ \om . 
\eeq 
The components of ${\mathfrak C}^{-1}_{mn}$ are one-forms. 
The key observation is that

\begin{proposition}
{\it For any Hamiltonian 
form $F$ 
we can construct a new Hamiltonian 
form $F'\approx F$,  
\beq \label{constrdir}
F' := F - 
\sigma \pbr{F}{\mathfrak{C}_m} \bullet ({\mathfrak C}^{-1}_{mn}\we {\mathfrak C}_n) ,  
\eeq
which has vanishing brackets with all second class constraints. }
\end{proposition}
Here, $\sigma $ is the signature of the metric whose appearence 
will be clarified later. 
The parentheses in Eq. (\ref{constrdir}) are important because the algebraic system of 
exterior forms equipped with two products $\we$ and $\bullet$ 
is {\bf not associative}. 

\begin{proof}
Using the (graded) Leibniz rule with respect to the 
$\bullet$-product, which is fulfilled by the Poisson-Gerstenhaber 
bracket of Hamiltonian forms, we obtain  
\beqa
\pbr{F'}{{\mathfrak C}_k} &=&\pbr{F}{{\mathfrak C}_k} - \sigma \pbr{\pbr{F}{\mathfrak{C}_m}\bullet 
({\mathfrak C}_{mn}^{-1}\we {\mathfrak C}_n)}{{\mathfrak C}_k} \nn \\
&=& \pbr{F}{{\mathfrak C}_k} 
- \sigma\pbr{F}{\mathfrak{C}_m}\bullet  \pbr{{\mathfrak C}_{mn}^{-1}\we {\mathfrak C}_n}{{\mathfrak C}_k} 
\nn\\
&-& \sigma \pbr{\pbr{F}{\mathfrak{C}_m}}{{\mathfrak C}_k}\bullet ({\mathfrak C}_{mn}^{-1}\we {\mathfrak C}_n) .  
\eeqa 
As the last term  weakly vanishes, 
\beq \label{egn22}
\pbr{F'}{{\mathfrak C}_k} 
\approx \pbr{F}{{\mathfrak C}_k} 
- \sigma \pbr{F}{\mathfrak{C}_m}\bullet  \pbr{{\mathfrak C}_{mn}^{-1}\we {\mathfrak C}_n}{{\mathfrak C}_k} . 
\eeq
%
Let us consider the bracket 
\beq
\pbr{{\mathfrak C}_k}{{\mathfrak C}_{mn}^{-1}\we {\mathfrak C}_n} = -\pounds_{X_{{\mathfrak C}_k}} ({\mathfrak C}_{mn}^{-1}\we {\mathfrak C}_n), 
\eeq
where $\pounds_{X_{{\mathfrak C}_k}}$ is the Lie derivative 
 with respect to the vector field 
$X_{{\mathfrak C}_k}$  
associated with the Hamiltonian form ${\mathfrak C}_k$ via the correspondence 
given by the polysymplectic form $\Omega$: 
$$X_{{\mathfrak C}_k}\inn\ \Omega = d {\mathfrak C}_k$$   
(c.f. Eq. (7)).  
As the Lie derivative is a derivation on the $\we$-algebra, 
\beq
 \pbr{{\mathfrak C}_k}{{\mathfrak C}_{mn}^{-1}\we {\mathfrak C}_n}= -\pounds_{X_{{\mathfrak C}_k}} ({\mathfrak C}_{mn}^{-1}) \we {\mathfrak C}_n 
- {\mathfrak C}_{mn}^{-1}\we \pounds_{X_{{\mathfrak C}_k}}({\mathfrak C}_n) , 
\eeq 
we obtain 
\beq 
 \pbr{{\mathfrak C}_k}{{\mathfrak C}_{mn}^{-1}\we {\mathfrak C}_n}
\approx  {\mathfrak C}_{mn}^{-1}\we \pbr{{\mathfrak C}_k}{{\mathfrak C}_n} 
= - {\mathfrak C}_{mn}^{-1} \we  {\mathfrak C}_{nk}
= - \delta_{mk}\ \omega .  
\eeq 
Further, using $*\om  = \sigma$,  
so that for any form $F$:  
$F\bullet  \omega 
= \sigma F, $  
we obtain  
\beq
\pbr{F'}{{\mathfrak C}_k} 
\approx \pbr{F}{{\mathfrak C}_k} 
-\sigma \pbr{F}{\mathfrak{{\mathfrak C}}_m}\bullet \delta_{mk} \omega 
= \pbr{F}{{\mathfrak C}_k}  - \sigma^2 \pbr{F}{{\mathfrak C}_k} = 0. 
\eeq
Hence, 
\[ 
\pbr{F'}{{\mathfrak C}_k} \approx 0  .
\] 
\end{proof}


Now, following a well known  way of introducing the Dirac bracket, 
we may try 
to define it as $ \pbr{F'}{G'}$. 
Though $F' \approx F$ and  $G' \approx G$, the bracket 
$\pbr{F'}{G'}$ is not weakly equal to $\pbr{F}{G}$  
and the desired 
properties 
\beq \label{props}
\pbr{F'}{G'}\approx \pbr{F'}{G}\approx \pbr{F}{G'}  
\eeq
are satisfied 
if $F$ and $G$ have the degree $(n-1)$.   
In this case we obtain the following formula for 
the analogue of the Dirac bracket  
\beq \label{diracbracket}
\pbr{F}{G}^D := \pbr{F}{G} - \sigma \pbr{F}{\mathfrak{C}_m}
\bullet ({\mathfrak C}_{mn}^{-1}\we \pbr{ {\mathfrak C}_n}{G}).
\eeq

\noindent 
The properties of the bracket (\ref{diracbracket}) are 
described by the following   
\begin{proposition} 
The Dirac bracket of any Hamiltonian 
(n-1)- 
form with a second class constraint  vanishes. 
\end{proposition}
\begin{proof}
The statement follows from Eq. (\ref{props}) and Proposition 3.2. 
\end{proof}

\begin{proposition} The Dirac bracket fulfills 
the Lie algebra properties 
\beqa 
\pbr{F}{G}^D & \approx & - \pbr{G}{F}^D , \nn \\
\pbr{\pbr{F}{G}^D}{K}^D &+&
 \pbr{\pbr{G}{K}^D}{F}^D + \pbr{\pbr{K}{F}^D}{G}^D  \approx 0  .
\eeqa
\end{proposition}
\begin{proof}
These weak identities follow from the properties of the Poisson-Gerstenhaber bracket in Theorem 2.1  
and Eq. (\ref{props}). 
\end{proof}

The space of Hamiltonian $(n-1)$ forms on which the bracket has been 
defined is not closed with respect the product operations $\bullet$ and $\we$. 
We, therefore, can not discuss the Leibniz property without extending 
the space of forms. A minimal extension would involve forms of degree $0$ 
or  $n$. Let us extend the definition of the bracket in Eq. (33) to the case when 
one of the arguments is a $0$-form $k$ 
(the bracket of two $0$-forms vanishes identically):
\beq
\pbr{F}{k}^D= \pbr{F}{k} - \sigma \pbr{F}{\mathfrak{C}_m}
\bullet ({\mathfrak C}_{mn}^{-1}\we \pbr{ {\mathfrak C}_n}{k}).
\eeq 
%
%
Then the following analogues of the Leibniz property follow:
\beqa 
\pbr{F}{kG}^D&=& \pbr{F}{k}^D G + k \pbr{F}{G}^D, \\ 
\pbr{F}{kl}^D&=& \pbr{F}{k}^D l + k \pbr{F}{l}^D .
\eeqa
The Leibniz properties for the $\bullet-$multiplication with $n$-forms  
are similar because $F\bullet k\omega = \sigma k F$. 


\section{Conclusion} 
We presented a  generalization of the Dirac constrained dynamics 
algorithm and a generalization of the Dirac bracket formula 
to the De Donder-Weyl Hamiltonian formalism of degenerate Lagranigian theories. 
The Dirac bracket is defined for dynamical variables given by Hamiltonian 
$(n-1)-$forms if the constraints can also be organized into Hamiltonian 
$(n-1)-$forms. A possible   generalization to forms of arbitrary degree 
and non-Hamiltonian forms  
remains an open issue. We also left beyond the scope of our discussion 
a possible geometrical interpetation of the construction in terms of the 
restriction of the polysymplectic structure to the subspace of constraints. 

Let us recall that in the geometric calculus of variations 
the issues related to the 
degeneracy  of the Lagrangian can be treated with the help 
of Lepagean equivalents of 
the Cartan form \cite{lepagean} which effectively modify the definitions of 
polymomentum variables and the corresponding regularity  conditions. 
It would be interesting to understand if our formal Dirac-like treatment 
of constraints in the 
DW theory could be understood geometrically using the theory 
of Lepagean equivalents. 
Yet another promising approach 
to the geometrical undertanding of the formalism could be based on the 
extension of the Cartan form to the constrained variational problems 
as presented in Ref. \cite{rodrigo}.  

Note in conclusion, that the approach developed here can be applied to the 
Dirac spinor field \cite{ikanat-tobe} which is a singular theory with  
second-class constraints from the point of view of the 
DW Hamiltonian formulation.

\section*{Acknowledgments} 
I thank Prof. Hellwig for his warm hospitality at TU Berlin. 
A discussion with  Prof. Krupkov\'a is gratefully acknowledged.


\begin{thebibliography}{99}


\bibliographystyle{ws-procs9x6}
 \bibliography{ws-pro-sample}

\bibitem{ikanat93} I.V.~Kanatchikov, 
On the canonical ctructure of the De Donder-Weyl covariant Hamiltonian 
formulation of field theory I. Graded Poisson brackets and the 
equations of motion, 
{\tt hep-th/9312162}. 

\bibitem{romp98} I.V.~Kanatchikov, 
Canonical structure of classical 
field theory in the polymomentum phase space,    
{\em Rep.~Math.~Phys. \bf 41} (1998) 49-90,   
{\tt hep-th/9709229}.  

\bibitem{bial96} I.V.~Kanatchikov, 
 On field theoretic generalizations of a Poisson algebra,   
{\em Rep. Math. Phys. \bf 40} (1997)  225-34,     
{\tt hep-th/9710069}. 

\bibitem{goslar21} I.V.~Kanatchikov, 
Novel algebraic structures from the polysymplectic form in field theory, 
{\tt hep-th/9612255}.



\bibitem{rund} 
Th. De Donder, 
{\em Th\'eorie Invariantive du Calcul des Variations, }  
 Gauthier-Villars, Paris (1935);  \\
H. Weyl, 
 {Geodesic fields in the calculus of variations, } 
 {\em Ann. Math. (2) \bf 36}, 607--29 (1935);  \\
H. Rund, 
{\em The Hamilton-Jacobi Theory in the Calculus of 
Variations, } D. van Nostrand, Toronto (1966).  

\bibitem{gimm} M.J. Gotay,   J. Isenberg and J. Marsden, 
{\sl Momentum  maps and classical relativistic fields\/},  
Part I:  Covariant field theory,   
{\tt physics/9801019},  
Part II: Canonical Analysis of Field Theories, {\tt math-ph/0411032}.  

\bibitem{henneaux} M. Henneaux and C. Teitelboim, 
{\sl Quantization of Gauge Systems, \/}
Princeton Univ. Press, Princeton, NJ 1992.  



\bibitem{bial97}  I.V. Kanatchikov,  
{De Donder-Weyl theory and a hypercomplex 
extension of quantum mechanics to field theory, } 
 {\em Rep. Math. Phys. \bf 43} (1999) 157-70,    
{\tt hep-th/9810165};   \\
I.V. Kanatchikov,  
On quantization of field theories in polymomentum variables,    
  {\tt hep-th/9811016}. 

\bibitem{ijtp2001} I.V. Kanatchikov, 
Precanonical quantum gravity: quantization without 
the space-time decomposition, 
{\em Int. J. Theor. Phys. \bf 40} (2001) 1121-49,  
{\tt gr-qc/0012074}.  
 
\bibitem{opava2001} I.V. Kanatchikov, 
Geometric (pre)quantization in the polysymplectic 
approach to field theory, 
{\tt hep-th/0112263}. 

\bibitem{ym} I.V. Kanatchikov, 
Precanonical quantization of Yang-Mills fields 
and the functional Schr\"odinger representation, 
 {\em Rep. Math. Phys. \bf 53} (2004) 181-193,    
{\tt hep-th/0301001}.



\bibitem{romanroy}
  N.~Roman-Roy, A.~M.~Rey, M.~Salgado and S.~Vilarino,
  On the k-Symplectic, k-Cosymplectic and Multisymplectic Formalisms of
  Classical Field Theories,
  {\tt arXiv:0705.4364 [math-ph].}

\bibitem{francav}
  M.~Francaviglia, M.~Palese and E.~Winterroth,
  A new geometric proposal for the Hamiltonian description of classical
  field theories,
 {\tt math-ph/0311018}.

\bibitem{paufler} M. Forger, C. Paufler, H. R\"omer, 
Hamiltonian multivector fields and poisson forms in multisymplectic 
field theory, 
{\em J. Math. Phys. \bf 46} (2005) 112901,   {\tt math-ph/0407057.}

\bibitem{deleon} M. de Le\'on, M. McLean, L. K. Norris, 
A. Rey-Roca, M. Salgado, 
Geometric structures in field theory, 
{\tt math-ph/0208036}; \\
  A.~Echeverria-Enriquez, M.~de Le\'on, M.~C.~Munoz-Lecanda and N.~Roman-Roy,
  Hamiltonian systems in multisymplectic field theories,
  {\tt arXiv:math-ph/0506003.}

\bibitem{krupkova} O. Krupkov\'a, Hamiltonian field theory, 
{\em J. Geom. Phys.  \bf 43}  (2002) 93-132.

\bibitem{sardan} G. Giachetta, L. Mangiarotti  
and  G. Sardanashvily,  
 {\em New Lagrangian and Hamiltonian Methods in Field Theory\/},  
 World Scientific,  Singapore 1997. 

\bibitem{helein} F. H\'elein, J. Kouneiher, 
Covariant Hamiltonian formalism for the calculus 
of variations with several variables, 
{\tt math-ph/0211046}. 


\bibitem{lepagean} O. Krupkov\'a and D. Smetanov\'a, 
Legendre transformation for regularizable Lagrangians in field theory, 
{\em Lett. Math. Phys. } {\bf 58} (2001) 189-204; \\
O. Krupkov\'a and D. Smetanov\'a, 
 On regularization of variational problems in first-order field theory,  
{\em  Rend. Circ. Mat. Palermo (2) Suppl.} No. 66 (2001), 133-140; \\
D. Smetanov\'a, Hamiltonian systems in dimension 4, 
to appear in {\em Proc. 10th Int. Conf.  Diff. Geom. and Its Applications, } 
Olomouc, Czech Republic (2007).  

\bibitem{rodrigo} A. Garcia, P.L. Garc\'\i a and C. Rodrigo, 
Cartan forms for first  order constrained variational problems, 
{\em J. Geom. Phys. \bf 56} (2006) 571-610. 


\bibitem{ikanat-tobe} I. Kanatchikov, in preparation. 

\end{thebibliography}
\end{document}